\let\oldvec\vec
\RequirePackage{fix-cm}
\documentclass[smallextended]{svjour3}       % onecolumn (second format)
\let\vec\oldvec

\smartqed  % flush right qed marks, e.g. at end of proof

\usepackage{hyperref}

\usepackage[ruled,vlined,linesnumbered,commentsnumbered]{algorithm2e}
\usepackage{amsmath,amssymb,amsfonts}
\usepackage{graphicx}

\newcommand {\ignore} [1] {}

\newcommand{\CC}{{\cal C}}
\newcommand{\FF}{{\cal F}}
\newcommand{\HH}{{\cal H}}
\newcommand{\RR}{{\cal R}}
\newcommand{\empt}{\emptyset}
\newcommand{\sem}{\setminus}
\newcommand{\subs}{\subseteq}

\newcommand{\f}{\frac}

\begin{document}

\title{On rooted $k$-connectivity problems in quasi-bipartite digraphs\thanks{Preliminary version in CSR 2021: 339-348.}}

\author{Zeev Nutov}

\institute{Z. Nutov \at The Open University of Israel \\
              % Tel.: +123-45-678910\\
              % Fax: +123-45-678910\\
              \email{nutov@openu.ac.il}}

\date{Received: date / Accepted: date}

\thispagestyle{empty}

% \noindent
% {\bf Conflict of Interest} \\
% On behalf of all authors, the corresponding author states that there is no conflict of interest.

% \medskip \medskip

% \noindent
% {\bf Funding} \\
% This paper received no funding. 

% \medskip \medskip

% \noindent
% {\bf Data availability} \\
% Data sharing not applicable to this article as no data sets were generated or analyzed during the current study.

\clearpage
\pagenumbering{arabic} 

\maketitle

%%%%%%%%%%%%%%%%%%%%%%%%%%%%%%%%%

\begin{abstract}
We consider the directed {\sc Min-Cost Rooted Subset $k$-Edge-Connection} problem: 
given a digraph $G=(V,E)$ with edge costs, a set $T \subs V$ of terminals, a root node $r$,
and an integer $k$, find a min-cost subgraph of $G$ that contains $k$ edge disjoint $rt$-paths for all $t \in T$.  
The case when every edge of positive cost has head in $T$ admits a polynomial time algorithm due to Frank \cite{F-R},
and the case when all positive cost edges are incident to $r$ is equivalent to the {\sc $k$-Multicover} problem. 
Chan et al. \cite{CLWZ} gave an LP-based $O(\ln k \ln |T|)$-approximation algorithm for quasi-bipartite instances, 
when every edge in $G$ has an end (tail or head) in $T \cup \{r\}$.
We give a simple combinatorial algorithm with the same ratio for a more general problem 
of covering an arbitrary $T$-intersecting supermodular set function 
by a minimum cost edge set, and for the case when only every positive cost edge has an end in $T \cup \{r\}$.
\end{abstract}


%%%%%%%%%%%%%%%%%%%%%%%%
\section{Introduction} \label{s:introduction}
%%%%%%%%%%%%%%%%%%%%%%%%

All graphs considered here are directed, unless stated otherwise.
We consider the following problem (a.k.a. {\sc $k$-Edge-Connected Directed Steiner Tree}):

\begin{center} \fbox{\begin{minipage}{0.97\textwidth}
\underline{{\sc Min-Cost Rooted Subset $k$-Edge-Connection}} \\
Input:  \ \ A directed (multi-)graph $G=(V,E)$ with edge costs $\{c(e):e \in E\}$, 
a set $T \subset V$ of terminals, a root node $r \in V \sem T$, and an integer $k$. \\
Output: A min-cost subgraph that has $k$ edge disjoint $rt$-paths for all $t \in T$. 
\end{minipage}} \end{center}

The case when every edge of positive cost has head in $T$ admits a polynomial time algorithm due to Frank \cite{F-R}.
When all positive cost edges are incident to $r$ we get the {\sc Min-Cost Multicover} problem. 
The case when all positive cost edges are incident to the same node admits approximation ratio $O(\ln n)$ \cite{KNs}.
More generally, a graph (or an edge set) 
is called {\bf quasi-bipartite} if every edge has at least one end (tail or head) in $T \cup \{r\}$. 

In the augmentation version of the problem -- {\sc Min-Cost Rooted Subset $(k_0,k)$-Edge-Connection Augmentation},
the input graph $G$ contains a subgraph $G_0=(V,E_0)$ of cost zero 
that has $k_0$-edge disjoint $rt$-paths for all $t \in T$.
Recently, Chan, Laekhanukit, Wei, \& Zhang \cite{CLWZ} obtained approximation ratio $O(\ln (k -k_0+1) \ln |T|)$ for the case
when $G$ is quasi-bipartite. We provide a simple proof for a more general setting.

An integer valued set function $f$ on a groundset $V$ is {\bf intersecting supermodular} 
if any $A,B \subs V$ that intersect satisfy the {\bf supermodular inequali\-ty}
$f(A)+f(B) \leq f(A \cap B)+f(A \cup B)$; if this holds whenever $A \cap B \cap T \neq \empt$ 
for a given set $T \subs V$ of terminals, then $f$ is {\bf $T$-intersecting supermodular}. 
We say that $A \subs V$ is an {\bf $f$-positive set} if $f(A)>0$.
$f$ is {\bf positively $T$-intersecting supermodular} if the supermodular inequality holds 
whenever $A \cap B \cap T \neq \empt$ and $f(A),f(B)>0$.
A typical way to create a positively intersecting supermodular function is to take the ``non-negative part'' of an intersecting
supermodular one, which means replacing each negative value by zero; namely, if $g$ is $T$-intersecting supermodular 
then $f(A)=\max\{g(A),0\}$ is positively $T$ intersecting supermodular, see \cite{F-R}.  

An edge $e$ {\bf covers} a set $A$ if it enters $A$, namely, if its head is in $A$ and tail is not in $A$.
For an edge set/graph $J$ let $d_J(A)$ denote the number of edges in $J$ that cover $A$.
We say that {\bf $J$ covers $f$} or that $J$ is a {\bf cover of $f$} if $d_J(A) \geq f(A)$ for all $A \subs V$.
We consider the following generic problem.

\begin{center} \fbox{\begin{minipage}{0.97\textwidth}
\underline{{\sc Min-Cost Set Function Edge Cover}} \\
{\em Input:}  \ \ A digraph $G=(V,E)$ with edge costs and a set function $f$ on $V$. \\
{\em Output:} A min-cost edge subset $J \subs E$ that covers $f$. 
\end{minipage}} \end{center}

Here $f$ may not be given explicitly, and for a polynomial time implementation of algorithms 
we need that certain queries related to $f$ can be answered in polynomial time. 
For an edge set $I$, the {\bf residual function $f^I$ of $f$} is defined by $f^I(A)=\max\{f(A)-d_I(A),0\}$. 
It is known that if $f$ is positively $T$-intersecting supermodular then so is $f^I$, c.f. \cite{F-R};
to see this, note that $g(A)=f(A)-d_I(A)$ is positively $T$ -intersecting supermodular 
(since $g(A) > 0$ implies $f(A)>0$ and since $-d(A)$ is supermodular), 
and thus the positive part $\max\{g(A),0\}$ of $g$ is also positively $T$-intersecting supermodular. 

Let $\max(f)=\max\{f(A):A \subs V\}$ denote the maximum $f$-value taken over all sets.
An inclusion minimal member of a set-family $\FF$ is called an {\bf $\FF$-core}, or simply a {\bf core}, 
if $\FF$ is clear from the context. Let $\CC_\FF$ denote the family of $\FF$-cores.
We will assume the following. 

\medskip 

\noindent
{\bf Assumption 1.}
{\em The cores of the set family $\FF=\{A: f^I(A)=\max(f^I)\}$
can be found in polynomial time for any edge set $I$.}

\medskip 

Given a set function $f$ on $V$ and a set $T \subs V$ of terminals, 
we say that a graph $G=(V,E)$ is {\bf $f$-quasi-bipartite} if every its edge  
has an end (tail or head) $v$ such that $v \in T$ or such that $v$ does not belong to any $f$-positive~set.
Let $E_0$ be the set of zero cost edges of $G$.
By Menger's Theorem, {\sc Min-Cost Rooted Subset $k$-Edge-Connection Augmentation} is equivalent
to the problem of finding a min-cost edge set $J \subs E \sem E_0$ that covers the function $f$ defined by 
\begin{equation*}
f(A)= \left\{ 
\begin{array}{ll} 
\max\{k-d_{G_0}(A),0\} \ & \mbox{ if } A \cap T \neq \empt, r \notin A \\ 
0                                        & \mbox{ otherwise }        
\end{array} 
\right .
\end{equation*}
This $f$ is positively $T$-intersecting supermodular, see \cite{F-R}.
Since $r$ does not belong to any $f$-positive set, 
if $G$ is quasi-bipartite then $G \sem E_0$ is $f$-quasi-bipartite. 
Assumption~1 holds for this $f$, since the cores as in Assumption~1 can be found
by computing for every $t \in T$ the closest to $t$ minimum $rt$-cut of $G_0+I$, c.f. \cite{F-R,N-p}.
Under Assumption~1, we prove the following. 

\begin{theorem} \label{t:1} 
The {\sc Min-Cost Set Function Edge Cover} problem with posi\-tively $T$-intersecting supermodular $f$ 
and $f$-quasi-bipartite $G$ admits approximation ratio $4H(\max(f)) \cdot (1+\ln |T|)$,
where $H(k)=\sum_{i=1}^k 1/i$ denotes the $k$th Harmonic number. 
\end{theorem}

Theorem~\ref{t:1} implies the following extension of the result of Chan et al. \cite{CLWZ}.

\begin{corollary} \label{c:r}
The {\sc Min-Cost Rooted Subset $(k_0,k)$-Edge-Connection Augmentation} problem 
admits approximation ratio $4H(k-k_0) \cdot (1+\ln|T|)$ if the set of positive cost edges of $G$ 
is quasi-bipartite.
\end{corollary}

As far as we can see, Corollary~\ref{c:r} cannot be deduced from the work of Chan et al. \cite{CLWZ}.
Our approach is motivated by an earlier result of Frank \cite{F-R}, 
who showed that {\sc Min-Cost Rooted Subset $k$-Edge-Connection}  
can be solved in polynomial time provided that every positive cost edge has head in $T$. 
For this, he proved that {\sc Min-Cost Set Function Edge Cover} with positively $T$-intersecting supermodular $f$ 
can be solved in polynomial time provided that every positive cost edge has head in $T$. 
While our approximation ratio is asymptotically similar to the one of \cite{CLWZ} -- $O(\ln k \cdot \ln |T|)$,
our constant hidden in the $O(\cdot)$ term is smaller and the proof (of a more general result) is substantially simpler.
Moreover, our algorithm is combinatorial and thus is much faster than the one of \cite{CLWZ},
that repeatedly solves linear programs and rounds LP solutions. % and applies randomized rounding on LP solutions. 
Chan et al. \cite{CLWZ} do not specify how the LPs are solved, but one can easily see
that they can be solved using the ellipsoid algorithm.  

We use a method initiated by the author in \cite{N-p}, 
that extends the Klein-Ravi \cite{KR} algorithm for the {\sc Node Weighted Steiner Tree} problem, 
to high connectivity problems.
It was applied later in \cite{N-nw,N-FOCS09} also for 
node weighted problems, and the same method is used in \cite{CLWZ}; 
a restricted version of this method appeared earlier in \cite{KN-cr} and later in \cite{FL}.
The method was further developed by Fukunaga \cite{Fuk} and 
Chekuri, Ene, and Vakilian \cite{CEV} for prize-collecting connectivity problems. 

In the rest of this section we briefly survey some literature on rooted connectivity problems.
The {\sc Directed Steiner Tree} problem admits approximation ratio $O(\ell^3|T|^{2/\ell})$ in time $O(|T|^{2\ell} n^\ell)$
for any integer $\ell$, see \cite{Z,CCCD,KP,HRZ},
and also a tight quasi-polynomial time approximation $O(\log^2 |T|/\log \log |T|)$ \cite{GLS,GN};
see also a survey in \cite{E-rg}. For similar results for {\sc Min-Cost Rooted Subset $2$-Edge-Connection} see \cite{GL}.
{\sc Directed Steiner Tree} is $\Omega(\log^2 n)$-hard to approximate even on very special instances \cite{HK} that arise
from the {\sc Group Steiner Tree} problem on trees; the latter problem admits 
a tight approximation ratio $O(\log^2 n)$ \cite{GKR}.
The (undirected) {\sc Steiner Tree} problem was also studied extensively, c.f. \cite{BGRS,GORZ} and the references therein.
The study of quasi-bipartite instances was initiated for undirected graphs in the 90's \cite{RV},  
while the directed version was shown to admits approximation ratio $O(\ln |T|)$ in \cite{FKS,HF}.

Rooted $k$-connectivity problems were studied
for both directed and undirected graphs,
edge-connectivity and node-connectivity, 
and various types of graphs and costs; c.f. a survey \cite{N-sn}.
For undirected graphs the problem admits approximation ratio $2$ \cite{Jain}, 
but for digraphs it has approximation threshold $\max\{k^{1/2-\epsilon},|T|^{1/4-\epsilon}\}$ \cite{L-hard}.
For the undirected node connectivity version, the currently best known approximation ratio is $O(k \ln k)$ \cite{N-FOCS09}
and threshold $\max\{k^{0.1-\epsilon},|T|^{1/4-\epsilon}\}$ \cite{L-hard}.
However, the augmentation version when any edge can be added by a cost of $1$
is just {\sc Set Cover} hard and admits approximation ratios $O(\ln |T|)$ for digraphs 
and $\min\{O(\ln |T|, O(\ln^2 k)\}$ for graphs \cite{KN-aug};
a similar result holds when positive cost edges form a star \cite{KNs}.

In digraphs, node connectivity can be reduced to edge-connectivity by a folklore reduction
of ``splitting'' each node $v$ into two nodes $v^{\sf in},v^{\sf out}$. 
However, this reduction does not preserve quasi-bipartiteness. 
The reductions of \cite{LN} that transfers undirected connectivity problems into directed ones,
and a reduction of \cite{CLNV} that reduces general connectivity requirements to rooted requirements,
also do not preserve quasi-bipartiteness.

%%%%%%%%%%%%%%%%%%%%%%%%%%%%%%%%%%%%%%%%%%%%%%%
\section{Covering {\em T}-intersecting supermodular functions (Theorem~\ref{t:1})} \label{s:t1}
%%%%%%%%%%%%%%%%%%%%%%%%%%%%%%%%%%%%%%%%%%%%%%%

A set family $\FF$ is a {\bf $T$-intersecting family} if 
$A \cap B,A \cup B \in \FF$ whenever $A \cap B \cap T \neq \empt$.
It is known that if $f$ is (positively) $T$-intersecting supermodular then 
the family $\FF=\{A \subseteq V:f(A)=\max(f)\}$ is $T$-intersecting, see \cite{F-R}.
We say that an edge set $I$ {\bf covers $\FF$} if $d_I(A) \geq 1$ for all $A \in \FF$.
Recall that inclusion minimal members of $\FF$ are called $\FF$-cores, 
and that $\CC_\FF$ denotes the family of $\FF$-cores. 
For $C \in \CC_\FF$ let $\FF(C)$ denote the family of sets in $\FF$ that contain no core distinct from $C$;
for $\CC\subs \CC_\FF$ let $\FF(\CC)=\cup_{C \in \CC} \FF(C)$. 

An analogue of the following lemma was proved in \cite[Lemma~3.3]{N-p} for intersecting families, 
and the proof for $T$-intersecting families is similar.

\begin{lemma} \label{l:d}
Let $\FF$ be a $T$-intersecting family.
If an edge set $S$ covers $\FF(\CC)$ for $\CC \subs \CC_\FF$ then $\nu(\empt)-\nu(S) \geq |\CC|/2$,
where $\nu(S)$ denotes the number of  cores of the residual family $\FF^S=\{A \in \FF:d_S(A)=0\}$.
\end{lemma}
\begin{proof}
The $\FF^S$-cores are $T$-disjoint, and each of them contains some $\FF$-core. 
Every $\FF^S$-core that contains a core from $\CC$ contains at least two $\FF$-cores. 
Thus the number of $\FF^S$-cores that contain exactly one $\FF$-core is at most $\nu(\empt)-|\CC|/2$.
Consequently, $\nu(S) \leq \nu(\empt)-|\CC|/2$.
\qed
\end{proof}

Consider an instance of the {\sc Min-Cost Set Function Edge Cover} problem with 
positively $T$-intersecting supermodular $f$ and $f$-quasi-bipartite $G$, 
and optimal solution value $\tau_f$.
Let $\FF=\{A \subs V:f(A)=\max(f)\}$, 
and for $I \subs E$ let $\nu_f(I)$ denote the number of $\FF^I$-cores. 
In the next section we will prove the following. 

\begin{lemma} \label{l:f}
There exists a polynomial time algorithm that 
finds $\empt \neq \CC \subs \CC_\FF$ and a cover $S \subs E$ of $\FF(\CC)$ such that 
$$
\f{c(S)}{|\CC|} \leq \f{2}{\max(f)} \cdot \f{\tau_f}{|\CC_\FF|}=\f{2}{\max(f)} \cdot \f{\tau_f}{\nu_f(\empt)} \ .
$$
\end{lemma}

Now let $I \subs E$ be an edge set such that $\nu_f(I) \geq 1$, and note that then $\max(f^I)=\max(f)$.
Applying Lemmas \ref{l:d} and \ref{l:f} on the residual function $g=f^I$ we get
that we can find in polynomial time an edge set $S \subs E \sem I$ such that 
$$
\f{c(S)}{\nu_g(\empt)-\nu_g(S)} \leq \f{c(S)}{|\CC|/2} \leq \f{4}{\max(g)} \cdot \f{\tau_g}{\nu_g(\empt)} \ .
$$
Observing that 
$\nu_g(\empt)=\nu(I)$, $\nu_g(S)=\nu_f(I \cup S)$, and $\tau_g \leq \tau_f$ we get:

\begin{corollary} \label{c:back'}
There exists a polynomial time algorithm that given $I \subs E$ with $\nu_f(I) \geq 1$
finds an edge set $S \subs E \sem I$ such that  
$$
\f{c(S)}{\nu_f(I)-\nu_f(I \cup S)} \leq \f{4}{\max(f)}\cdot \f{\tau_f}{\nu_f(I)} \ .
$$
\end{corollary}

From Corollary~\ref{c:back'} it is a routine to deduce the following corollary, 
c.f. \cite{KR} and \cite[Theorem~3.1]{N-nw}; we provide a proof for completeness of exposition.

\begin{corollary} \label{c:back}
There exists a polynomial time algorithm that computes a cover $I$ of $\FF=\{A \subs V:f(A)=\max(f)\}$ 
of cost $c(I) \leq \f{4}{\max(f)} \cdot  (1+\ln \nu_f(\empt)) \cdot \tau_f$. 
\end{corollary}
\begin{proof}
Start with $I=\empt$ an while $\nu_f(I) \geq 1$ add to $I$ an edge set $S$ as in Corollary~\ref{c:back'}.
Let $I_j$ be the partial solution at the end of iteration $j$, where $I_0 =\empt$, and let $S_j$ be the set
added at iteration $j$; thus $I_j = I_{j-1} \cup S_j$, $j = 1,\ldots,q$. 
Let $\nu_j = \nu_f(I_j)$, so $\nu_0 = \nu_f(\empt)$, $\nu_q=0$, and $\nu_{q-1} \geq 1$. 
Let $\rho=\f{4}{\max(f)}$. Then 
$$\f{c_j}{\nu_{j-1}-\nu_j} \leq \rho \cdot \f{\tau_f}{\nu_{j-1}} \ \ \ \ \ j=1, \ldots, q \ . $$ 
This implies $c_q \leq \rho \tau_f$
and 
$$\nu_j \leq \nu_{j-1} \left(1-\f{c_j}{\rho \tau_f}\right) \ \ \ \ \ j=1, \ldots,q \ .$$
Unraveling we get
$$
\f{\nu_{q-1}}{\nu_0} \leq \prod_{j=1}^{q-1}\left(1-\f{c_j}{\rho \tau_f} \right) \ .
$$
Taking natural logarithms and using the inequality $\ln(1 + x) \leq x$, we obtain
$$
\rho \cdot \tau_f \cdot \ln \left( \f{\nu_0}{\nu_{q-1}} \right) \geq \sum_{j=1}^{q-1} c_j \ .
$$
Since $c_q \leq \rho \tau_f$ and $\nu_{q-1} \geq 1$, we get
$c(I) \leq c_q +\sum_{j=1}^{q-1} c_j \leq \rho \tau_f(1+\ln \nu_0)$.
\qed
\end{proof}

To see that Corollary~\ref{c:back} implies Theorem~\ref{t:1}, consider the following algorithm
that uses the so called ``backward augmentation'' method. % due to \cite{GGPS}. 

\medskip

\begin{algorithm}[H]
\caption{{\sc Backward-Augmentation$(f,G=(V,E),c)$}} 
\label{alg:AC} 
$I \gets \empt$ \\
\For{$\ell=\max(f)$ {\em downto} $1$}
{Compute a cover $I_\ell$ of $\FF_\ell=\{A \subs V:f^I(A)=\ell\}$ as in Corollary~\ref{c:back} \\
$I \gets I \cup I_\ell$}
\Return{$I$}
\end{algorithm}

\medskip

At iteration $\ell$ we have $c(I_\ell)/\tau_f  \leq  4(1+\ln |T|)/\ell$, hence
the overall approxi\-mation ratio is $4(1+\ln |T|) \cdot \sum_{\ell=\max(f)}^1 1/\ell=4H(\max(f)) \cdot (1+\ln |T|)$, 
as required in Theorem~\ref{t:1}. 
It remains only to prove Lemma~\ref{l:f}, which is done in the next section,
where we also describe a simple polynomial time implementation of our algorithm. 

%%%%%%%%%%%%%%%%%%%
\section{Proof of Lemma~\ref{l:f}}
%%%%%%%%%%%%%%%%%%%

Let $\langle G=(V,E),c,T,f \rangle$ be an instance of {\sc Min-Cost Set Function Edge Cover} with 
positively $T$-intersecting supermodular $f$ and $f$-quasi-bipartite $G$, 
and an optimal solution value $\tau=\tau_f$. Let us denote $p=\max(f)$ and let $\FF=\{A \subs V:f(A)=p\}$.
Recall that $\FF(C)$ denotes the family of sets in $\FF$ that contain no core distinct from $C$, 
and that $\FF(\CC)=\cup_{C \in \CC} \FF(C)$ for $\CC\subs \CC_\FF$.
We need to show that there exists a subfamily of cores $\CC \subs \CC_\FF$ and a cover $S \subs E$ of $\FF(\CC)$ such that 
\begin{equation} \label{e:S}
\f{c(S)}{|\CC|} \leq \f{1}{p} \cdot \f{\tau}{|\CC_\FF|} \ .
\end{equation}
We also need to design a polynomial time algorithm that finds such $\CC,S$. 

%%%%%%%%%%%%%%%%%%%%%%
\subsection{\bf Roadmap of the proof} \label{ss:roadmap}
%%%%%%%%%%%%%%%%%%%%%%

Here is a roadmap of the proof of Lemma~\ref{l:f}. 
To make this roadmap a complete proof we just need to  describe a polynomial time implementation
and to prove formally three Lemmas \ref{l:A}, \ref{l:B}, and \ref{l:D} mentioned in this roadmap; 
this is done in Sections \ref{ss:implementation} and \ref{ss:proofs}, respectively.

We say that $I \subs E$ is a {\bf $p$-cover} of $\FF$ if $d_I(A) \geq p$ for all $A \in \FF$, and 
$I$ is {\bf $\FF$-quasi-bipartite} if every edge in $I$ 
has an end (tail or head) $v$ such that $v \in T$ or such that $v$ does not belong to any set in $\FF$.
Fix an optimal solution $I \subs E$, so $I$ is a cover of $f$ of cost $c(I)=\tau$.
Note that $I$ is a $p$-cover of $\FF$ (since $f(A)=p$ for all $A \in \FF$) and that 
$I$ is $\FF$-quasi-bipartite (since $G$ is $f$-quasi-bipartite and since $I \subs E$).

\begin{figure} \centering 
\includegraphics{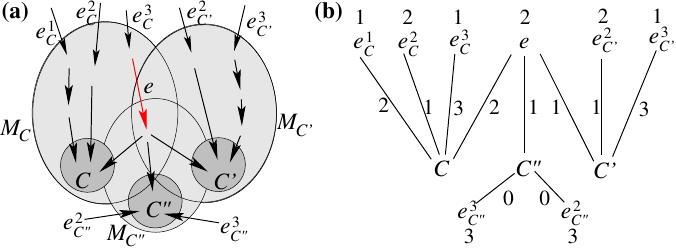}
\caption{% Illustration to the construction of $\HH$ for $p=3$. 
(a) A $3$-cover $I$ of $\FF(\{C,C',C''\})$; 
here $e_{C''}^2,e_{C''}^3$ have cost $3$ each, $e,e_C^2,e_{C'}^2$ have cost $2$ each, 
and all other edges have cost $1$.
(b) The auxiliary graph $\HH$.
The star  $S_\HH$ with center $e$ and leaf set $\CC=\{C,C',C''\}$ has ratio $\f{c(S_\HH)}{|\CC|}=\f{6}{3}=2$ 
(the same ratio $2$ is achieved by the star $S_\HH \sem \{C\}$).
The edge subset $S$ of $I$ that corresponds to $S_\HH$ is $I_C^3 \cup I_{C'}^1 \cup I_{C''}^1$.
Here $c(I)=26$ and $c(\HH)=29$. 
Note that $e_{C'}^1=e=e_{C''}^1$, and that $e \in I_C^3$ but $e \neq e_C^3$.}
\label{f:prime} \end{figure}

\begin{enumerate}
\item[{\bf (A)}]
For every $C \in \CC_\FF$ fix some inclusion minimal $p$-cover $I_C \subs I$ of $\FF(C)$.
In {\bf Lemma \ref{l:A}} we show the following:
\begin{itemize}
\item[(i)]
Each $I_C$ partitions into $p$ inclusion minimal $1$-covers $I_C^1,\ldots, I_C^p$ of $\FF(C)$.
\item[(ii)]
Each $\FF(C)$ has a unique inclusion maximal set $M_C$ 
and each $I_C^j$ has a unique edge $e_C^j$ that covers $M_C$, which we call the {\bf prime edge of $I_C^j$}.
\end{itemize}
\item[{\bf (B)}]
In {\bf Lemma~\ref{l:B}} we show that for distinct $C,C' \in \CC_\FF$ and any $1 \leq j,j' \leq p$,
if $I_C^j \cap I_{C'}^{j'} \neq \empt$ then $I_C^j \cap I_{C'}^{j'}=\{e_C^j\}$ or $I_C^j \cap I_{C'}^{j'}=\{e_{C'}^{j'}\}$,
see Fig.~\ref{f:prime}(a);
this property is since $I$ is $\FF$-quasi-bipartite.
Consequently, for every $e \in I$ there is at most one set $I_C^j$ such that $e \in I_C^j$ and $e \neq e_C^j$.  
\item[{\bf (C)}]
Construct an auxiliary bipartite graph $\HH$ with node- and edge-costs as follows, see Fig.~\ref{f:prime}(b)
The node parts of $\HH$ are the prime edges and $\CC_\FF$.
Each node $e$ of $\HH$ that is a prime edge inherits its cost $c(e)$ in $G$,
and is connected to each $C \in \CC_\FF$ such that $e \in I_C^j$ for some $j$ 
by an edge of cost $c(I_C^j)-c(e)$ (this edge represents the set $I_C^j$).
Since for every $e \in I$ at most one set $I_C^j$ contains $e$ as a non-prime edge,
and since the sets $I_C^j$ are pairwise disjoint,
the total cost of $\HH$ is at most $2$ times the cost of $I$. 
\item[{\bf (D)}]
Every node $C \in \CC_\FF$ of $\HH$ has at least $p$ neighbors in $\HH$ (the prime edges of the sets $I_C^1,\ldots, I_C^p$). 
In {\bf Lemma \ref{l:D}} we show that $\HH$ contains a star $S_\HH$ with leaf set $\CC \subs \CC_\FF$ such that 
$\f{c(S_\HH)}{|\CC|} \leq \f{1}{p} \cdot \f{c(\HH)}{|\CC_\FF|} \leq \f{2}{p} \cdot \f{\tau}{|\CC_\FF|}$.
Then the edge subset $S \subs I$ that corresponds to $S_\HH$ covers $\FF(\CC)$, 
and $S,\CC$ satisfy inequality (\ref{e:S}).
% Specifically, $S$ is obtained by taking for each 
\item[{\bf (E)}]
To find $\empt \neq \CC \subs \CC_\FF$ and a cover $S$ of $\FF(\CC)$ 
that satisfies (\ref{e:S}), we make a similar construction: 
now  $\HH$ has node set $E \cup \CC$,
every node $e \in E$ of $\HH$ has cost equal to the cost of $e$ in $G$,
and in $\HH$ each node $C \in \CC_\FF$ is connected to each node $e \in E$ by an edge of cost
being the minimum cost of an edge set $S$ such that $S \cup \{e\}$ covers $\FF(C)$.
In such a graph $\HH$ we can find a star $S_\HH$ with leaf set $\CC$ that minimizes $\f{c(S_\HH)}{|\CC|}$
using the method of Klein \& Ravi \cite{KR}; see also step 3 of the implementation discussed in the next section. 
\end{enumerate}

%%%%%%%%%%%%%%%%%%%%%%%%%%%%%%%%
\subsection{\bf Implementation} \label{ss:implementation}
%%%%%%%%%%%%%%%%%%%%%%%%%%%%%%%%

Here we briefly discuss a simple implementation of the entire algorithm.
We start with the particular case of the 
{\sc Min-Cost Rooted Subset $(k_0,k)$-Edge-Connection Augmentation} problem.
In what follows let $n=|V|$ and $m=|E|$. 
As a pre-processing step, we assign unit capacities to edges in $E$ and compute a $k_0$-flow from the root $r$ to each $t \in T$.
This can be done in $O(km|T|)$ time using the Ford-Fulkerson algorithm.
Let us consider iteration $\ell$ of Algorithm~\ref{alg:AC}, when $\max(f)=k-\ell$.
We will assume that we already have a flow on zero cost edges of value $k-\ell-1$ to each $t \in T$,
and perform the following steps. 
\begin{enumerate}
\item
We increase the flow by $1$ to each $t \in T$, and discard terminals for 
which the flow can be further increased by $2$. This can be done in $O(m|T|)$ time. 
\item
To compute the cost of an edge of $\HH$ between nodes $C$ and $e$,
we add a ``dummy'' edge of cost $0$ from $r$ to some terminal in every core distinct from $C$, 
set the cost of $e$ to $0$, 
and compute a minimum cost edge set that increases the $rC$-flow by $1$; 
the later problem admits a linear time reduction to the shortest path problem and thus can be 
implemented in $O(n^2)$ time. 
The number of edges in $\HH$ is $O(m|T|)$, hence 
$\HH$ can be constructed in $O(n^2m|T|)$ time. 
\item
We can sort the edges of $\HH$ by increasing cost in $O(m|T| \log n)$ time. 
Then finding a (nontrivial) star $S^e$ in $\HH$ with a specific center $e$ that minimizes $\f{c(S^e)}{|\CC|}$
can be done in time linear in the degree of $e$ in $\HH$ as follows. 
We take the lowest cost edge incident to $e$ into $S^e$ and then add edges incident to $e$ one by one 
in increasing cost order until reaching a local minimum of  $\f{c(S^e)}{|\CC|}$; see \cite{KR}. 
The overall time for computing all stars $S^e$ is $O(m n \log n)$, 
which is dominated by the time $O(n^2m|T|)$ of the construction of $\HH$. 
\item
At iteration $\ell$ we need to construct the graph $\HH$
at most $|T|$ times, hence the overall time per iteration $\ell$ is $O(n^2 m |T|^2)$. 
And since we have $k-k_0$ iterations, the overall running time is $(k-k_0) \cdot O(n^2 m|T|^2)=O(kn^6)$.
\end{enumerate}
We note that while the running time of the described implementation is somewhat high, 
it is still much lower than that of Chan et al. \cite{CLWZ}.

The implementation of steps $1,3,4$ for the {\sc Min-Cost Set Function Edge Cover} problem under Assumption~1 is similar.
For step~2, for any $C \in \CC_\FF$ and $e \in E\sem I$ 
we need to find in polynomial time a min-cost edge set $S=S(e,C)$ such that $S \cup \{e\}$ covers $\FF(C)$.
For this, it is sufficient to find a min-cost cover of $\FF(C)$ after resetting the cost of $e$ to zero. 
The family $\FF(C)$ is a $T$-intersecting family that has a unique core; such a family is called a {\bf ring}. 
It is known that a min-cost edge-cover of a ring can be found in polynomial time 
under Assumption~1 (c.f. \cite{F-R,N-p}), by a standard primal dual algorithm.

%%%%%%%%%%%%%%%%%%%%%%%%%%%%%
\subsection{\bf Proofs of Lemmas} \label{ss:proofs}
%%%%%%%%%%%%%%%%%%%%%%%%%%%%%

Now we turn to formal proofs of Lemmas \ref{l:A},\ref{l:B} and \ref{l:D} mentioned in our roadmap. 
At each step we will specify the part of our roadmap that is proved.

A $T$-intersecting family $\RR$ that has a unique core $C$ is called a {\bf ring}.
Then $C$ is the intersection of all sets in $\RR$, 
and $\RR$ also has a unique inclusion maximal set $M$ which is the union of all sets in $\RR$. 
The following lemma is a folklore.

\begin{lemma} \label{l:M}
If $\FF$ is a $T$-intersecting family then $\FF(C)$ is a ring family for any $C \in \CC_\FF$;
thus $\FF(C)$ also has a unique  inclusion maximal set $M_C$. 
Furthermore, $M_C \cap M_{C'} \cap T=\empt$ for any distinct $C,C' \in \CC_\FF$.
\end{lemma}

The next lemma gives two additional known properties of rings;  c.f. 
\cite{F-ker} for the first property and \cite[Lemma 2.6 and Corollary 2.7]{N-p} for the second.
These two properties imply part {\bf (A)}.

\begin{lemma}  \label{l:A}
Let $\RR$ be a ring with minimal member $C$ and maximal member $M$.
\begin{itemize}
\item[{\em (i)}]
Any $p$-cover of $\RR$ is a union of $p$ edge disjoint covers of $\RR$. 
\item[{\em (ii)}]
Let $I$ be an inclusion minimal cover of $\RR$. 
Then there is an ordering $e_1,e_2,\ldots,e_q$ of $I$ and a nested family  
$C=C_1 \subset C_2 \cdots \subset C_q = M$ of sets in $\RR$ such that 
for every $j=1,\ldots, q$, $e_j$ is the unique edge in $I$ that enters $C_j$ 
(namely, $e_j$ has head in $C_j$ and tail not in $C_j$).
\end{itemize}
\end{lemma}

Lemmas  \ref{l:M} and \ref{l:A}(i) imply the following lemma that implies parts {\bf (B,C)}.

\begin{lemma} \label{l:B}
Let $I$ be an $\FF$-quasi-bipartite cover of a $T$-intersecting family~$\FF$.
For $C \in \CC_\FF$ let $I_C \subs I$ be an inclusion minimal cover of $\FF(C)$,
and let $e_C$ be the unique (by Lemma~\ref{l:A}{\em (i)}) edge in $I_C$ that covers $M_C$.
Let $C,C' \in \CC_\FF$ be distinct and let $e \in I_C \cap I_{C'}$. Then $e=\{e_C\}$ or $e=\{e_{C'}\}$. 
\end{lemma}
\begin{proof}
Suppose that $e \neq e_C$ and we will show that then $e=e_{C'}$. 
Note that $e$ does not cover $M_C$, hence $e$ has both ends in $M_C$, 
by the minimality of $I_C$ and Lemma~\ref{l:A}(ii). 
Since $I$ is $\FF$-quasi-bipartite, $e$ has an end $t$ in $M_C \cap T$. 
By Lemma~\ref{l:M}, $t \notin M_{C'}$, hence by the minimality of $I_{C'}$ we must have $e=e_{C'}$. 
\qed
\end{proof}

The next lemma implies part {\bf (D)}.

\begin{lemma} \label{l:D}
Let $H=(A \cup B,E)$ be a bipartite graph with 
edge- and node- costs $\{c(e):e \in E\} \cup \{c(a):a \in A\}$
and let ${\cal S}$ be the set of stars in $H$ with center in $A$ and leaves in $B$. 
If the degree of every $b \in B$ is at least $p$ then there is $S^* \in {\cal S}$ such that   
$\f{c(S^*)}{|L(S^*)|} \leq \f{1}{p} \cdot \f{c(G)}{|B|}$, where $L(S^*)$ is the set of leaves of $S^*$.  
\end{lemma}
\begin{proof}
For $S \in {\cal S}$ let $c_S$ denote the cost of $S$ and let ${\bf c}=\{c_S:S \in {\cal S}\}$
be a vector of costs of the stars. For an integer $q$ let ${\cal L}(q)$ be the following set of linear constraints:
\[
\begin{array}{ll}
\displaystyle{\sum_{L(S) \ni b} x_S \geq q}  & \ \ \ \forall  b \in B  \\
0 \leq x_S \leq 1                                         & \ \ \ \forall S \in {\cal S}
\end{array} 
\]
Note that the characteristic vector ${\bf x}$ of the inclusion maximal stars in ${\cal S}$
satisfies the set of constraints ${\cal L}(p)$ and that ${\bf c} \cdot {\bf x}=c(H)$. 
Thus the vector ${\bf y}={\bf x}/p$ satisfies ${\cal L}(1)$ and ${\bf c} \cdot {\bf y}=c(H)/p$.
Let $S^*=\arg\max_{S \in {\cal S}} \f{|L(S)|}{c(S)}$. Then 
$$
\f{|L(S^*)|}{c(S^*)} ({\bf c} \cdot {\bf y})    \geq 
\sum_{S \in {\cal S}} \f{|L(S)|}{c_S} c_S y_S=\sum_{S \in {\cal S}} |L(S)|y_S = \sum_{b \in B} \sum_{L(S) \ni b} y_S \geq \sum_{b \in B} 1 =|B| \ .
$$
The first inequality is by the choice of $S^*$ and the second inequality is since ${\bf y}$ 
satisfies ${\cal L}(1)$.

From this we get that $\f{|L(S^*)|}{c(S')} \geq \f{|B|}{{\bf c} \cdot {\bf y}}$, so 
$\f{c(S^*)}{|L(S^*)|} \leq \f{{\bf c} \cdot {\bf y}}{|B|}=\f{{\bf c} \cdot {\bf x}/p}{|B|}=\f{1}{p} \cdot \f{c(H)}{|B|}$. 
\qed
\end{proof}

\medskip 

This concludes the proof of Lemma~\ref{l:f}, and thus also the proofs
Theorem~\ref{t:1} and Corollary~\ref{c:r} are complete.

% \bibliographystyle{plain}
% \bibliography{aug-bib}

\end{document}